\documentclass[11pt]{article}

\usepackage{mathpazo}
\usepackage{amsfonts}
\usepackage{amsmath}
\usepackage{graphicx}
\usepackage{latexsym}
\usepackage{mathabx}

\usepackage[margin=1.05in]{geometry}
  
\usepackage[T1]{fontenc}
\usepackage{times}
\usepackage{color,graphicx}
\usepackage{array}
\usepackage{enumerate}
\usepackage{amsmath}
\usepackage{amssymb}
\usepackage{amsthm}
\usepackage{pgfplots}
\usepackage{pgf}
\usepackage{tikz}
\usetikzlibrary{patterns}
\usepgfplotslibrary{patchplots} 
\usetikzlibrary{pgfplots.patchplots} 
\pgfplotsset{width=9cm,compat=1.5.1}

\usepackage{amsthm}

\newtheorem{definition}{Definition}
\newtheorem{lemma}{Lemma}
\newtheorem{theorem}{Theorem}
\newtheorem{corollary}{Corollary}

\newtheorem{example}{Example}

\usepackage{xcolor}
\usepackage{makeidx}
\usepackage[colorlinks=true,linkcolor=blue,anchorcolor=blue,citecolor=red,urlcolor=magenta]{hyperref}

\newcommand{\tr}{\operatorname{Tr}}
\newcommand{\rank}{\operatorname{rank}}

\newcommand{\bra}[1]{\langle #1 |}
\newcommand{\ket}[1]{| #1 \rangle}

\newcommand{\ketbra}[2]{| #1 \rangle\langle #2 |}

\newcommand{\defeq}{\stackrel{\smash{\textnormal{\tiny def}}}{=}}

\begin{document}
\title{Birkhoff--James Orthogonality in the Trace Norm, with Applications to Quantum Resource Theories}

\author{
	Nathaniel Johnston,\textsuperscript{\!\!1,2} \ Shirin Moein,\textsuperscript{\!\!1,2,3} \ Rajesh Pereira,\textsuperscript{\!\!2} \ and Sarah Plosker\textsuperscript{2,3}
}

\maketitle

\begin{abstract}
	We develop numerous results that characterize when a complex Hermitian matrix is Birkhoff--James orthogonal, in the trace norm, to a (Hermitian) positive semidefinite matrix or set of positive semidefinite matrices. For example, we develop a simple-to-test criterion that determines which Hermitian matrices are Birkhoff--James orthogonal, in the trace norm, to the set of all positive semidefinite diagonal matrices. We then explore applications of our work in the theory of quantum resources. For example, we characterize exactly which quantum states have modified trace distance of coherence equal to $1$ (the maximal possible value), and we establish a connection between the modified trace distance of $2$-entanglement and the NPPT bound entanglement problem.\\
	
	\noindent \textbf{Keywords:} Trace norm, Birkhoff--James orthogonality, quantum coherence, quantum entanglement, quantum resource theory\\
	
	\noindent \textbf{MSC2010 Classification:} 
	15A03, 
15A60, 
15B57, 
81P40
\end{abstract}

\addtocounter{footnote}{1}
\footnotetext{Department of Mathematics \& Computer Science, Mount Allison University, Sackville, NB, Canada E4L 1E4}\addtocounter{footnote}{1}
\footnotetext{Department of Mathematics \& Statistics, University of Guelph, Guelph, ON, Canada N1G 2W1} \addtocounter{footnote}{1}
\footnotetext{Department of Mathematics \& Computer Science, Brandon University, Brandon,
    MB, Canada R7A 6A9}

\section{Introduction}\label{sec:intro}
Birkhoff--James orthogonality was introduced in \cite{james1947orthogonality} to provide a definition of orthogonality in normed vector spaces that extends the usual one from inner product spaces. Necessary and sufficient conditions for Birkhoff--James orthogonality of matrices in the operator norm were given in \cite{bhatia1999orthogonality}, and a study of  Birkhoff--James orthogonality of matrices in the Schatten $p$-norms as well as operator norm was performed in \cite{li2002orthogonality}. Birkhoff--James orthogonality of a given Hermitian matrix to every member of the subspace of real diagonal matrices, under the operator norm, was studied in \cite{ALRV12}.

In this work, we characterize when a given complex Hermitian matrix is Birkhoff--James orthogonal in the trace norm to a given (Hermitian) positive semidefinite matrix. We then explore numerous consequences of this main result, such as a simple-to-check characterization of when a given Hermitian matrix is Birkhoff--James orthogonal in the trace norm to every positive semidefinite diagonal matrix.

Our result has several natural applications in quantum information theory, where (mixed) quantum states, also known as density matrices, are positive semidefinite trace-one matrices. When working with a quantum resource theory, it is natural to ask how far a given quantum state is from a particular convex set of states of interest \cite{Reg18}, and the natural norm to use to measure distance is the trace norm \cite[Chapter 3]{Wat18}. For example, in the resource theory of quantum coherence \cite{BCP14}, which is of particular interest in quantum optics, quantum biology, and quantum thermodynamics \cite{Glau63,Su63}, the aforementioned convex set of states is exactly those that are diagonal. In this case, our results concerning Birkhoff--James orthogonality to all positive semidefinite diagonal matrices provide a characterization of states that are a trace distance of $1$ (the largest such distance possible) from that set, thus extending a result from \cite{johnston2018modified} from pure states to mixed states.

We also extend our results to several other quantum resource theories, including those of entanglement and $k$-coherence. In particular, when applying our result to the resource theory of $2$-entanglement (i.e., Schmidt number $2$ \cite{TH00}), we show that the long-standing NPPT bound entanglement problem \cite{HRZ20} has a natural rephrasing in terms of Birkhoff--James orthogonality.

\subsection{Notation and Terminology}\label{sec:notation}

We now introduce the mathematical preliminaries that we need to discuss and prove our results concerning Birkhoff--James orthogonality. We defer a brief introduction to quantum information theory to Section~\ref{sec:quantum_resource_theories}, when we need it.

Let $M_n$ be the set of all $n\times n$ matrices with complex entries and let $M_n^+$ be the subset of them that are (Hermitian) positive semidefinite. We use bold lower case letters such as $\mathbf{v}$ and $\mathbf{w}$ to denote vectors in $\mathbb{C}^n$, with the entries of $\mathbf{v}\in \mathbb C^n$ denoted by $v_1$, $v_2$, $\ldots$, $v_n$. We denote the standard basis vectors by $\{\mathbf{e_j}\}_{j=1}^n$, and we use $A_{i,j}$ to denote the $(i,j)$-entry of a matrix $A\in M_n$. The eigenvalues of $A$ are denoted by $\lambda_j$ or by $\mu_j$, while the singular values of $A$---the non-negative square roots of the eigenvalues of $AA^*$---are denoted by $\sigma_1 \geq \sigma_2 \geq \cdots \geq \sigma_n \geq 0$. If $A$ and $B$ are Hermitian matrices, we use the notation $A \succeq B$ (or $B \preceq A$) to denote $A-B$ being positive semidefinite).

The trace norm of $A$ is $\|A\|_{\textup{tr}} \defeq \sum_{j=1}^n \sigma_j$ and the operator norm of $A$ is $\|A\| \defeq \sigma_1$. If $A$ is Hermitian then its singular values are the absolute values of its eigenvalues, so $\|A\|_{\textup{tr}} = \sum_{j=1}^n |\lambda_j|$. When discussing norms in general, we use the notation $\vvvert\cdot\vvvert$, to avoid confusion with the operator norm. Given a Hermitian matrix $A$, we denote the  eigenprojection matrices, projecting onto the direct sum of the eigenspaces corresponding to the positive, zero, and negative eigenvalues, respectively, by $P^+$, $P_0$, and $P^-$, and we note that $P^++P_0+P^-=I$. For simplicity of terminology, we refer to these as the orthogonal projections onto the positive, zero, and negative eigenspaces of $A$ from now on.

As a generalization of orthogonality of vectors in a Hilbert space, we have the following notion of orthogonality in Banach spaces \cite{james1947orthogonality}:

\begin{definition}\label{defn:bj_orth}
    Suppose $(\mathcal{X}, \vvvert\cdot\vvvert)$ is a Banach space over $\mathbb{R}$, and $A, B \in \mathcal{X}$. We say that $A$ is \emph{Birkhoff--James orthogonal} to $B$ if
    \[
        \vvvert A\vvvert \leq \vvvert A + \lambda B\vvvert \quad \text{for all} \quad \lambda \in \mathbb{R}.
    \]
\end{definition}

We note that this notion of orthogonality is homogeneous and additive, but not symmetric: if $A$ is Birkhoff--James orthogonal to $B$, then it is not necessarily the case that $B$ is Birkhoff--James orthogonal to $A$.

Here, we are interested in the case where a Hermitian matrix $A$ is Birkhoff--James orthogonal under the trace norm $\|\cdot\|_{\textup{tr}}$ to a positive semidefinite matrix $B$. For this reason, from now on we simply use the terminology \emph{Birkhoff--James orthogonal} to signify Birkhoff--James orthogonal \emph{with respect to the trace norm}, and we note that $A$ being Birkhoff--James orthogonal to $B$ means exactly that
\[
    \|A\|_{\textup{tr}} \leq \|A+\lambda B\|_{\textup{tr}} \quad \text{for all} \quad \lambda \in \mathbb{R}.
\]

Definition~\ref{defn:bj_orth} of Birkhoff--James orthogonality works just fine if the field $\mathbb{R}$ is replaced by $\mathbb{C}$. For now (and only now), we briefly consider Birkhoff--James orthogonality in this sense, and recall the following result:

\begin{theorem}\cite[Theorem 3.3]{li2002orthogonality}\label{LS1} Let $A,B\in M_n$. Then the following are equivalent.
\begin{itemize}
  \item[(a)]\label{a:Birkhoff-James orthogonal}
 $A$ is Birkhoff--James orthogonal to $B$ (in the trace norm).
    \item[(b)]\label{b:Birkhoff-James orthogonal}
There exists a  matrix $M\in M_n$ with   $\|M\|\leq 1$ 
such that $\tr(AM^*)=\Vert A\Vert_{\textup{tr}}$  and $\tr(BM^*)=0$.
\end{itemize}
\end{theorem}

While this result in this form is stated for the complex vector space $M_n$, it applies straightforwardly to the real vector space of Hermitian matrices as well. Indeed, if $A$ is Hermitian and $\tr(AM^*)$ is real, then
\[
    \tr(AM^*)=\tr(AM)=\tr\left(A\left(\frac{M+M^*}{2}\right)\right),
\]
and hence $M$ can be chosen to be Hermitian in this case. We thus only consider Birkhoff--James orthogonality in the \emph{real} vector space of Hermitian matrices from now on.

It is perhaps worth noting that Theorem~\ref{LS1} can be proved in a different way than was done in \cite{li2002orthogonality}: via the subgradient \cite[Theorem 2]{Watson}, which is a useful tool in convex optimization. Indeed, the same is true of many of our upcoming results, but we avoid using the subgradient in this paper, opting instead for more elementary proofs.

\subsection{Arrangement of the Paper}\label{sec:arrangement}

In Section~\ref{sec:main_results}, we present and prove our main result, which characterizes when a Hermitian matrix is Birkhoff--James orthogonal to a positive semidefinite one. In Section~\ref{sec:diagonal}, we apply our main result to the special case of Birkhoff--James orthogonality to diagonal matrices. In particular, we characterize which Hermitian matrices are Birkhoff--James orthogonal to all positive semidefinite diagonal matrices, and we briefly consider the problem of Birkhoff--James orthogonality with the set of \emph{all} (not necessarily positive semidefinite) diagonal matrices.

In Section~\ref{sec:quantum_resource_theories}, we introduce quantum resource theories and show how our results characterize which quantum states are ``most resourceful'' in the sense of the modified trace distance of that resource. Furthermore, in Section~\ref{sec:nppt_be} we show that Birkhoff--James orthogonality of certain matrices is equivalent to a long-standard conjecture from quantum information theory, which states that a certain quantum state is ``bound entangled''. Finally, we close in Section~\ref{sec:conclusions} with some conclusions and open questions related to our work.

\section{Main Result}\label{sec:main_results}

We now explore the question of when a Hermitian matrix is Birkhoff--James orthogonal to a given positive semidefinite matrix. First, however, we need the following well-known inequality that bounds the trace of a product of Hermitian matrices (see \cite[Problem~III.6.14]{Bha97}, for example).

\begin{lemma}\label{firstlem}
    Suppose $H, M \in M_n$ are two Hermitian matrices with eigenvalues $\lambda_1\ge \lambda_2 \ge \cdots \ge \lambda_n$ and $\mu_1\ge \mu_2\ge \cdots \ge \mu_n$, respectively. Then
    \[
        \tr(HM) \le \sum_{j=1}^{n}\lambda_j\mu_j,
    \]
    with equality if and only if there exists an orthonormal basis $\{\mathbf{v_j}\}_{j=1}^n \subset \mathbb{C}^n$ such that, for all $j$, $\mathbf{v_j}$ is an eigenvector of both $H$ and $M$ corresponding to $\lambda_j$ and $\mu_j$, respectively.
\end{lemma}

We also need one more lemma before we will be able to state and prove our main result:

\begin{lemma}\label{secondlem}
    Suppose $H, M \in M_n$ are two Hermitian matrices, and let $P^+$ and $P^-$ denote the orthogonal projections onto the strictly positive and strictly negative eigenspaces of $H$, respectively. Then $2P^+-I\preceq M\preceq I-2P^-$ if and only if $\|M\|\le 1$ and $\tr(HM)=\Vert H\Vert_{\textup{tr}}$.
\end{lemma}

\begin{proof}
    Let $\lambda_1\ge \lambda_2 \ge \cdots \ge \lambda_n$ and $\mu_1\ge \mu_2\ge \cdots \ge \mu_n$ denote the eigenvalues of $H$ and $M$, respectively. Then $\|M\|\le 1$ if and only if $-1\le \mu_j \le 1$ for all $j$. If this condition holds then Lemma~\ref{firstlem} tells us that
    \[
        \tr(HM)\le \sum_{j=1}^{n}\lambda_j\mu_j\le \sum_{j=1}^{n}\vert \lambda_j \vert=\Vert H\Vert_{\textup{tr}}.
    \]
    We thus have $\tr(HM)=\Vert H\Vert_{\textup{tr}}$ if and only if both of the previous two inequalities are equalities. The first inequality is an equality if and only if there exists an orthonormal basis $\{\mathbf{v_j}\}_{j=1}^n \subset \mathbb{C}^n$ such that, for all $j$, $\mathbf{v_j}$ is an eigenvector of both $H$ and $M$ corresponding to $\lambda_j$ and $\mu_j$, respectively. The second inequality is an equality if and only if $\mu_j=1$ whenever $\lambda_j>0$ and $\mu_j=-1$ whenever $\lambda_j<0$. These two conditions together are equivalent to $2P^+-I\preceq M\preceq I-2P^-$, as claimed.
\end{proof}

We now have enough machinery to state and prove our main result:

\begin{theorem}\label{thm:main_res}
    Suppose $H,B \in M_n$ are Hermitian and $B$ is positive semidefinite. Let $P^+$ and $P^-$ denote the orthogonal projections onto the strictly positive and strictly negative eigenspaces of $H$, respectively. Then $H$ is Birkhoff--James orthogonal to $B$ in the trace norm if and only if both
    \[
        \tr(BP^+)\le \frac{1}{2}\tr(B) \quad \text{and} \quad \tr(BP^-)\le \frac{1}{2}\tr(B).
    \]
\end{theorem}

\begin{proof}
    Theorem \ref{LS1} and Lemma~\ref{secondlem} together imply that  $H$ is Birkhoff--James orthogonal to $B$ in the trace norm if and only if there exists a Hermitian matrix $M$ with $2P^+-I\preceq M\preceq I-2P^-$ and $\tr(BM)=0$. These conditions imply that both $\tr(B(I-2P^+))\ge 0$ and $\tr(B(I-2P^-))\ge 0$, which is equivalent to $\tr(BP^+)\le \frac{1}{2}\tr(B)$  and $\tr(BP^-)\le \frac{1}{2}\tr(B)$ by the linearity of the trace.
    
    Conversely, suppose we have both $\tr(BP^+)\le \frac{1}{2}\tr(B)$ and $\tr(BP^-)\le \frac{1}{2}\tr(B)$. If either $\tr(B(I-2P^+))=0$ or $\tr(B(I-2P^-))=0$, we choose $M=2P^+-I$ or $M=I-2P^-$, respectively. On the other hand, if we have $\alpha := \tr(B(I-2P^+)) > 0$ and $\beta := \tr(B(I-2P^-)) > 0$, we choose
    \[
        M=\frac{1}{\alpha+\beta}\big(\beta(2P^+-I)+\alpha(I-2P^-)\big).
    \]
    In either case, we then have $2P^+-I\preceq M\preceq I-2P^-$ and $\tr(BM)=0$, which means that $H$ is Birkhoff--James orthogonal to $B$ in the trace norm.
\end{proof}

\section{Birkhoff--James Orthogonality to Diagonal Matrices}\label{sec:diagonal}

We now consider the problem of determining which Hermitian matrices are Birkhoff--James orthogonal to every diagonal matrix. This problem was considered for the operator norm in \cite{ALRV12}, and was solved in small dimensions there, whereas we consider the trace norm version of it. We start with a result that determines when a Hermitian matrix is Birkhoff--James orthogonal  in the trace norm to all \emph{positive semidefinite} diagonal matrices:

\begin{theorem}\label{thm:optimal_ek_rephrase}
     Suppose $H \in M_n$ is Hermitian, and let $P^+$ and $P^-$ be the orthogonal projections onto its strictly positive and negative eigenspaces, respectively. Then the following are equivalent:
     
     \begin{itemize}
         \item[(a)] $H$ is Birkhoff--James orthogonal in the trace norm to every positive semidefinite diagonal matrix.
         
         \item[(b)] $H$ is Birkhoff--James orthogonal in the trace norm to $\mathbf{e_j}\mathbf{e}_{\mathbf{j}}^*$ for all $1 \leq j \leq n$.
         
         \item[(c)] $P^{+}_{j,j} \leq 1/2$ and $P^{-}_{j,j} \leq 1/2$ for all $1 \leq j \leq n$.
     \end{itemize}
\end{theorem}

\begin{proof}
    The equivalence of (b) and (c) follows from choosing $B = \mathbf{e_j}\mathbf{e}_{\mathbf{j}}^*$ in Theorem~\ref{thm:main_res}. The equivalence of (a) and (c) similarly follows from choosing $B$ to be an arbitrary diagonal positive semidefinite matrix scaled so that $\tr(B) = 1$ in Theorem~\ref{thm:main_res}.
\end{proof}

We now start looking at the more difficult problem of characterizing which Hermitian matrices are Birkhoff--James orthogonal to \emph{all} (not necessarily positive semidefinite) diagonal matrices. Theorem~\ref{thm:optimal_ek_rephrase} provides a necessary condition, but it is not a sufficient one when $n \geq 3$, as demonstrated by the following example:

\begin{example}\label{exam:higher_dim_bj_orth}
    The Hermitian matrix
    \[
        H = \begin{bmatrix}
        -1 &  5 & 2 \\
         5 & -1 & 2 \\
         2 &  2 & 2
        \end{bmatrix}
    \]
    has projections onto its negative and positive eigenspaces
    \[
        P^{-} = \frac{1}{2}\begin{bmatrix}
            1 & -1 & 0 \\ -1 & 1 & 0 \\ 0 & 0 & 0
        \end{bmatrix} \quad \text{and} \quad P^{+} = \frac{1}{3}\begin{bmatrix}
            1 & 1 & 1 \\ 1 & 1 & 1 \\ 1 & 1 & 1
        \end{bmatrix},
    \]
    respectively. It follows that $P_{j,j}^{-} \leq 1/2$ and $P_{j,j}^{+} \leq 1/2$ for all $1 \leq j \leq 3$, so Theorem~\ref{thm:optimal_ek_rephrase} tells us that $H$ is Birkhoff--James orthogonal to every positive semidefinite diagonal matrix.
    
    However, $H$ is not Birkhoff--James orthogonal to \emph{every} diagonal matrix. To see this, we can compute $\|H\|_{\textup{tr}} = 12$, but if $D = \mathrm{diag}(6,6,-6/5)$ then $\|H + D\|_{\textup{tr}} = 54/5 < 12$ (and in fact, semidefinite programming can be used to show that this $D$ is optimal, so $H + D$ is Birkhoff--James orthogonal to every diagonal matrix).
\end{example}

On the other hand, if $n = 2$ then the condition of Theorem~\ref{thm:optimal_ek_rephrase} is both necessary \emph{and} sufficient for Birkhoff--James orthogonality to every diagonal matrix, as we will see shortly. Our starting point towards proving this fact is the following simple corollary, which solves this problem for \emph{invertible} Hermitian matrices:

\begin{corollary}\label{cor:trace_min_even_dim}
    Suppose $H \in M_n$ is Hermitian and invertible, and let $P^+$ and $P^-$ be the orthogonal projections onto its positive and negative eigenspaces, respectively. The following are equivalent:
    \begin{itemize}
        \item[(a)] $H$ is Birkhoff--James orthogonal in the trace norm to $\mathbf{e_j}\mathbf{e}_{\mathbf{j}}^*$ for all $1 \leq j \leq n$.
        
        \item[(b)] $H$ is Birkhoff--James orthogonal in the trace norm to every diagonal matrix.
        
        \item[(c)] $P^{+} = (I + U)/2$ and $P^{-} = (I-U)/2$ for some Hermitian unitary $U$ with zeros on its diagonal.
    
        \item[(d)] $P_{j,j}^{+} = P_{j,j}^{-} = 1/2$ for all $1 \leq j \leq n$.
    \end{itemize}
\end{corollary}

\begin{proof}
    To see that (a) and (d) are equivalent, recall from Theorem~\ref{thm:optimal_ek_rephrase} that $H$ is Birkhoff--James orthogonal to each $\mathbf{e_j}\mathbf{e}_{\mathbf{j}}^*$ if and only if $P_{j,j}^{+} \leq 1/2$ and $P_{j,j}^{-} \leq 1/2$ for all $1 \leq j \leq n$. Since $H$ is invertible, $P^{+} + P^{-} = I$, so $P_{j,j}^{+} + P_{j,j}^{-} = 1$ for all $j$. It follows that $P_{j,j}^{+} \leq 1/2$ and $P_{j,j}^{-} \leq 1/2$ is equivalent to $P_{j,j}^{+} = P_{j,j}^{-} = 1/2$, so (a) is equivalent to (d).
    
    The fact that (c) implies (d) is trivial. To see that (d) implies (c), notice that (d) implies we can write $P^{+} = I/2 + X$ for some matrix $X$ with diagonal entries equal to $0$. Since $P^{+}$ is Hermitian, so is $X$, and since $P^{+}$ is a projection, we have $(I/2 + X)^2 = I/2 + X$, so $X^2 = I/4$, so $U := X/2$ is unitary. Then $P^{+} = (I + U)/2$, and since $P^{+} + P^{-} = I$ we must have $P^{-} = (I-U)/2$.
    
    The fact that (b) implies (a) is trivial. To complete the proof, we show that (d) implies (b) as follows. If (d) holds and $B$ is diagonal, then it must be the case that $\tr(BP^{+}) = (1/2)\tr(B)$ and $\tr(BP^{-}) = (1/2)\tr(B)$. It follows from Theorem~\ref{thm:main_res} that $H$ is Birkhoff--James orthogonal to $B$ in the trace norm, completing the proof.
\end{proof}

\begin{theorem}\label{thm:b_j_orthog_2dim}
     A Hermitian matrix $H \in M_2$ is Birkhoff--James orthogonal to \emph{every} diagonal matrix if and only if it is Birkhoff--James orthogonal to $\mathbf{e_1}\mathbf{e}_{\mathbf{1}}^*$ and $\mathbf{e_2}\mathbf{e}_{\mathbf{2}}^*$.
\end{theorem}

\begin{proof}
    The ``only if'' direction is trivial, and Corollary~\ref{cor:trace_min_even_dim} gives the result when $H$ is invertible, so we only prove the ``if'' direction in the case when $\rank(H) = 1$.
    
    In this case, $H$ has the form $H = \pm \mathbf{v}\mathbf{v}^*$ for some $\mathbf{v} \in \mathbb{C}^2$ (without loss of generality, we will assume that $H = \mathbf{v}\mathbf{v}^*$). Then $P^{+} = \mathbf{v}\mathbf{v}/\|\mathbf{v}\|^2$, and we see that if $P_{j,j}^{+} \leq 1/2$ for $j = 1,2$, then in fact we must have $P_{j,j}^{+} = 1/2$ for $j = 1,2$. It follows that
    \[
        H = \frac{\|\mathbf{v}\|^2}{2}\begin{bmatrix}
            1 & e^{i\theta} \\ e^{-i\theta} & 1
        \end{bmatrix}
    \]
    for some $\theta \in \mathbb{R}$. We can then choose $M = e^{i\theta}\mathbf{e_1}\mathbf{e}_{\mathbf{2}}^* + e^{-i\theta}\mathbf{e_2}\mathbf{e}_{\mathbf{1}}^*$ in Theorem~\ref{LS1} to see that $H$ is Birkhoff--James orthogonal to every diagonal matrix.
\end{proof}

We now present several corollaries of Theorem~\ref{thm:optimal_ek_rephrase} that place restrictions on which Hermitian matrices are Birkhoff--James orthogonal to the diagonal matrices $\mathbf{e_j}\mathbf{e}_{\mathbf{j}}^*$ ($1 \leq j \leq n$). To start, we note that these conditions bound the inertia of such matrices quite strongly.

\begin{corollary}\label{cor:trace_min_inertia_bounds}
    Suppose $H \in M_n$ is Hermitian and Birkhoff--James orthogonal to $\mathbf{e_j}\mathbf{e}_{\mathbf{j}}^*$ for each $1 \leq j \leq n$, and let $\mu_+$ and $\mu_-$ denote how many eigenvalues it has that are positive and negative, respectively. Then
    \[
        \mu_- \leq n/2 \quad \text{and} \quad \mu_+ \leq n/2.
    \]
\end{corollary}

\begin{proof}
     Theorem~\ref{thm:optimal_ek_rephrase} tells us that the  positive and negative eigenprojections $P^+$ and $P^-$ satisfy $P^{+}_{j,j} \leq 1/2$ and $P^{-}_{j,j} \leq 1/2$ for all $1 \leq j \leq n$. Taking the trace, we obtain $\sum_{j=1}^n P^+_{j,j}=\mu_+\leq n/2$, and similarly for $\mu_-$.
\end{proof}

\begin{corollary}\label{cor:trace_min_zero_eigs}
    Suppose $H \in M_n$ is Hermitian and $n$ is odd. If $H$ is Birkhoff--James orthogonal to $\mathbf{e_j}\mathbf{e}_{\mathbf{j}}^*$ for each $1 \leq j \leq n$ then it must have $0$ as an eigenvalue.
\end{corollary}

\begin{proof}
     Let $\mu_+$, $\mu_-$, and $\mu_0$ denote how many eigenvalues $H$ has that are positive, negative, and $0$, respectively. Suppose (for the sake of establishing a contradiction) that $\mu_0 = 0$. Then $\mu_++\mu_0+\mu_-=n$ implies that $\mu_++\mu_-=n$. However, Corollary~\ref{cor:trace_min_inertia_bounds} implies, since $n$ is odd, that $\mu_+ \leq (n-1)/2$ and similarly for $\mu_-$, which yields $\mu_++\mu_-\leq n-1<n$, a contradiction.
\end{proof}

The above Corollary~\ref{cor:trace_min_zero_eigs} really is specific to matrices of odd size: in the even-dimensional case, matrices that are Birkhoff--James orthogonal to every diagonal matrix do not necessarily have $0$ as an eigenvalue. For example, Theorem~\ref{thm:b_j_orthog_2dim} tells us that the invertible matrix
\[
    A = \begin{bmatrix}
        0 & 1 \\ 1 & 0
    \end{bmatrix}
\]
is Birkhoff--James orthogonal to every diagonal matrix.

We now present a \emph{necessary and sufficient} condition for a Hermitian matrix to be Birkhoff--James orthogonal to every diagonal matrix (in contrast with Theorem~\ref{thm:optimal_ek_rephrase}, which provides a necessary condition). 

\begin{corollary}\label{cor:general_res}
    Suppose $H \in M_n$ is Hermitian and let $P^+$, $P^-$, $P_0$ denote the orthogonal projections onto the strictly positive eigenspaces, strictly negative eigenspaces and null space of $H$ respectively. Then $H$ is Birkhoff--James orthogonal to every diagonal matrix if and only if there exists $X\in M_n$ with $X$ Hermitian and $-I\le X\le I$ such that the matrix $M=P^+-P^-+P_0XP_0$ has all of its main diagonal entries equal to zero.
\end{corollary}

\begin{proof} We note that if $-I\le X\le I$, then $2P^+-I=P^+-P^--P_0\le P^+-P^-+P_0XP_0 \le P^+-P^-+P_0=I-2P^-$. Hence by Lemma~\ref{secondlem}, the existence of an $X\in M_n$ with $X$ Hermitian and $-I\le X\le I$ and $M=P^+-P^-+P_0XP_0$ has all of its main diagonal entries equal to zero is equivalent to the existence of a Hermitian $M$ having all of its main diagonal entries equal to zero and satisfying $\| M \| =1$ and $\tr(HM)=\| H\|_{\textup{tr}}$ which by Theorem \ref{LS1} is equivalent to $H$ being Birkhoff--James orthogonal to every diagonal matrix.
\end{proof}

We note that any choice of Hermitian $U$ in any polar decomposition of $H$ corresponds to $U=P^+-P^-+P_0XP_0$ where $X$ is chosen to be Hermitian and unitary; further is $H$ is invertible then $P_0=0$ and $P^+-P^-$ is the unitary in the unique polar decomposition of $H$. The following theorem is a consequence of these observations together with Corollary \ref{cor:general_res}.

\begin{theorem}\label{thm:suff_cond}
    Suppose $H \in M_n$ is a Hermitian matrix. If there is a polar decomposition $H=UP$ in which the unitary $U$ is Hermitian and has all of its diagonal entries equal to $0$, then $H$ is Birkhoff--James orthogonal to every diagonal matrix. If $H$ is invertible then the converse also holds.
\end{theorem}

After some work, it can be seen that Theorem~\ref{thm:suff_cond} is essentially equivalent to a special case of \cite[Theorem 1.1]{grover2017orthogonality}.
It is perhaps worth recalling that if $H$ is invertible then its polar decomposition is unique, so the necessary and sufficient condition of Theorem~\ref{thm:suff_cond} is easy to check in this case. 

\subsection{When the Given Matrix is Positive Semidefinite}

In addition to the $2$-dimensional case that we saw in Theorem~\ref{thm:b_j_orthog_2dim}, there are some other extra conditions that we can add to $H$ that ensure that the (easy-to-check) properties of Theorem~\ref{thm:optimal_ek_rephrase} are equivalent to $H$ being Birkhoff--James orthogonal to every diagonal matrix. In particular, positive semidefiniteness of $H$ lets us increase the dimension a bit:

\begin{theorem}\label{thm:psd_birkhoff_orth_small_dim_low_rank}
    Suppose $H \in M_n^{+}$ is positive semidefinite and at least one of the following conditions hold: (a) $n \leq 4$, (b) $\rank(H) \in \{1,n/2,n\}$, or (c) the orthogonal projection onto $\mathrm{range}(H)$ has constant diagonal. Then $H$ is Birkhoff--James orthogonal to \emph{every} diagonal matrix if and only if it is Birkhoff--James orthogonal to $\mathbf{e_j}\mathbf{e}_{\mathbf{j}}^*$ for all $1 \leq j \leq n$.
\end{theorem}

\begin{proof}
    Since the ``only if'' implication of this theorem is trivial, we only prove the ``if'' direction.
    
    We start with the rank-1 case of part (b). In this case, notice that in this case we can write $H = \mathbf{v}\mathbf{v}^*$. By the same argument used in the proof of \cite[Theorem~2]{johnston2018modified}, there exists a matrix $M = (\mathbf{v}\mathbf{v}^* - \mathbf{w}\mathbf{w}^*)/\|\mathbf{v}\|^2$ with all of its diagonal entries equal to $0$, $|w_j| = |v_j|$ for all $1 \leq j \leq n$, and $\mathbf{v}^*\mathbf{w} = 0$ (this matrix $M$ was called $Y$ in the proof of \cite[Theorem~2]{johnston2018modified}). For an arbitrary diagonal matrix $D$, we thus have $\tr(DM^*)=0$, $\|M\|\leq 1$, and $\tr(HM^*) = \mathbf{v}^*M\mathbf{v} = 1 = \Vert H\Vert_{\textup{tr}}$, so Theorem~\ref{LS1} tells us that $H$ is Birkhoff--James orthogonal to $D$.
    
    To see that part (c) of the theorem holds, let $c$ be the constant diagonal value (i.e., if $P$ is the orthogonal projection onto $\mathrm{range}(H)$ then $c = P_{j,j}$ for all $j$). If we let $M=P/c-I$, then by Lemma~\ref{secondlem} we have $\|M\|\le 1$ and $\tr(HM)=\Vert H\Vert_{\textup{tr}}$. Since all diagonal entries of $M$ are zero, $\tr(MD)=0$ for all diagonal matrices. It then follows from Theorem~\ref{LS1} that $H$ is Birkhoff--James orthogonal to every diagonal matrix.
    
    The rank $n/2$ case of part (b) then holds because, in this case we must have $\tr(P)=\mathrm{rank}(H) = n/2$. When $H$ is Birkhoff--James orthogonal to $\mathbf{e_j}\mathbf{e}_{\mathbf{j}}^*$ for all $1 \leq j \leq n$, we must have all diagonal entries of $P$ exactly equal to $1/2$ to also satisfy both the trace condition and condition (c) of Theorem~\ref{thm:optimal_ek_rephrase}. Part (c) of this theorem (which we already proved) then implies the result.
    
    The rank $n$ case of part (b) follows immediately from Corollary~\ref{cor:trace_min_even_dim}.
    
    Finally, to see that part (a) of the theorem holds, notice that if $n = 2$ then Theorem~\ref{thm:b_j_orthog_2dim} gives the result. If $n = 3$ then Corollary~\ref{cor:trace_min_inertia_bounds} tells us that $\mathrm{rank}(H) = \mu_+ \leq n/2 = 3/2$, so $\mathrm{rank}(H) = 1$, and we already proved the rank-$1$ case. If $n = 4$ then Corollary~\ref{cor:trace_min_inertia_bounds} tells us that $\mathrm{rank}(H) = \mu_+ \leq n/2 = 2$. Regardless of whether $\mathrm{rank}(H) = 1$ or $\rank(H) = 2 = n/2$, case(b) of this theorem shows that we are already done.
\end{proof}

We note that Theorem~\ref{thm:psd_birkhoff_orth_small_dim_low_rank} is in a sense tight: if all three of the conditions (a)--(c) fail, then the conclusion of the theorem needs not hold, as demonstrated by the next example that has $n = 5$ and $1 < \rank(H) = 2 < n/2$.

\begin{example}\label{exam:n5_rank2}
    Let $\alpha = \sqrt{(\sqrt{10}-1)/3}$ and let $P \in M_5^{+}$ be the orthogonal projection onto the subspace
    \[
        \mathrm{span}\big( (1,1,1,0,0), (0,0,1,\alpha,1/\alpha) \big).
    \]
    Then the following claims are all straightforward to verify:
    
    \begin{itemize}
        \item[1)] $P$ is positive semidefinite with rank $2$;
        
        \item[2)] $P_{j,j} \leq 1/2$ for all $j$ (with $P_{5,5} = 1/2$ exactly), so Theorem~\ref{thm:optimal_ek_rephrase} tells us that it is Birkhoff--James orthogonal to every positive semidefinite diagonal matrix. However;
        
        \item[3)] If $D = \mathrm{diag}(0, 0, 3, -1, 3)$ then $\|P - D/40\|_{\textup{tr}} \approx 1.99441235 < 2 = \|P\|_{\textup{tr}}$, so $P$ is not Birkhoff--James orthogonal to $D$.
    \end{itemize}
\end{example}

\section{Applications to Quantum Resource Theories}\label{sec:quantum_resource_theories}

In quantum information theory, a \emph{mixed quantum state} or \emph{density matrix} is a positive semidefinite trace-one matrix, typically denoted by $\rho \in M_n$ or $\sigma \in M_n$. A \emph{pure quantum state} is a unit vector in $\mathbb{C}^n$, which we denote using ``bra-ket'' notation: $\ket{v} \in \mathbb{C}^n$ is a unit column vector, while $\bra{v} \defeq \ket{v}^*$ is the corresponding (dual) row vector. From now on, whenever we use lowercase Greek letters like $\rho$, we are assuming that it is a quantum state normalized to have $\tr(\rho) = 1$, and similarly if we use $\ket{v}$ then we are assuming it is a \emph{unit} vector. For un-normalized matrices and vectors, we denote them like $A \in M_n^+$ and $\mathbf{v} \in \mathbb{C}^n$, just like in the earlier sections of this paper.

A quantum resource theory \cite{CG18} consists of two things: (1) a closed convex cone of positive semidefinite matrices $C \subseteq M_n^{+}$, and (2) a set of linear maps that send $C$ back to itself (for our purposes, only (1) is relevant). The quantum states $\rho \in C$ (i.e., the members of $C$ with trace $1$) are called ``free states'', and are thought of as the states that are ``useless'' for some given task in quantum information theory or quantum computation.

The most well-known resource theory is that of \emph{entanglement}, where $C = C_{\textup{ent}}$ is the set of \emph{separable} quantum states \cite{GT09,HHH09,Wer89}:
\begin{align*}
    C_{\textup{ent}} & \defeq \left\{ \sum_j X_j \otimes Y_j : X_j \in M_m^{+}, Y_j \in M_n^{+} \ \text{for all} \ j\right\} \subseteq (M_m \otimes M_n)^{+}.
\end{align*}
However, numerous other resource theories are studied as well \cite{Reg18}. For example, another widely-used resource theory is that of \emph{coherence} \cite{Abe06,BCP14,Glau63,Su63}, in which the set of free states $C = C_{\textup{coh}}$ consists of exactly those that are diagonal (in the standard basis):
\begin{align*}
    C_{\textup{coh}} & \defeq \left\{ \sum_j x_j \mathbf{e_j}\mathbf{e}_{\mathbf{j}}^* : 0 \leq x_j \in \mathbb{R} \ \text{for all} \ j \right\} \subseteq M_n^{+}.
\end{align*}
In this context, the members of $C_{\textup{coh}}$ are called \emph{incoherent}.

Given a closed convex cone $C$, we can define the following \emph{modified trace distance of $C$}, which simply measures the distance from a given quantum state to $C$:
\begin{align}\label{eq:dist_defn}
    D_{C}(\rho) & \defeq \min\big\{ \|\rho - X\|_{\textup{tr}} : X \in C \big\}.
\end{align}
This quantity is a ``proper'' (i.e., physically relevant) measure of the resource specified by $C$ in the sense of \cite{Reg18}. When $C = C_{\textup{coh}}$, it is the \emph{modified trace distance of coherence} \cite{yu2016alternative}, and when $C = C_{\textup{ent}}$, it is the \emph{modified trace distance of entanglement} \cite{RFWG19}. For notational convenience, we omit the central level of subscripting when denoting this distance. For example, we denote the modified trace distance of coherence simply by $D_{\textup{coh}}$, rather than $D_{C_\textup{coh}}$.

As a bit of a historical note, we comment that an analogous (non-modified) \emph{trace distance of $C$} was originally defined in the $C = C_{\textup{coh}}$ case in \cite{rana2016trace} via the minimization
\[
    \min\big\{ \|\rho - X\|_{\textup{tr}} : X \in C, \tr(X) = 1 \big\}.
\]
However, this quantity was subsequently shown to \emph{not} be a ``proper'' measure of coherence (i.e., it lacked certain physically desirable properties that a quantification of a quantum resource should satisfy). The ``modified'' trace distance of $C$ was then introduced to fix this problem \cite{yu2016alternative}.

By applying our theorems concerning Birkhoff--James orthogonality (in particular, Theorem~\ref{thm:main_res}) in this setting, we get the following characterization of which mixed quantum states have $D_C(\rho) = 1$. These states are of interest because they are exactly those that are ``most resourceful'' or ``most useful'': it is clear from the definition of Equation~\eqref{eq:dist_defn} that it is always the case that $D_C(\rho) \leq 1$.

\begin{theorem}\label{thm:mod_tr_dist_resource}
    Suppose $C \subseteq M_n^{+}$ is a closed convex cone, $\rho \in M_n^{+}$ is a mixed quantum state, and $P$ is the orthogonal projection onto $\mathrm{range}(\rho)$. The following are equivalent:
    \begin{enumerate}
        \item[(a)] $D_C(\rho) = 1$.
        
        \item[(b)] $\rho$ is Birkhoff--James orthogonal in the trace norm to every member of $C$.
        
        \item[(c)] $\tr(P\sigma) \leq 1/2$ for every mixed quantum state $\sigma \in C$.
    \end{enumerate}
\end{theorem}

\begin{proof}
    To see that (a) is equivalent to (b), note (by definition) that $\rho$ is Birkhoff--James orthogonal in the trace norm to every member of $C$ if and only if
    \[
        1 = \|\rho\|_{\textup{tr}} \leq \|\rho + \lambda X\|_{\textup{tr}}
    \]
    for all $\lambda \in \mathbb{R}$ and $X \in C$. However, we can safely ignore positive values of $\lambda$ since adding the positive semidefinite matrix $\lambda X$ to the positive semidefinite matrix $\rho$ can only increase its eigenvalues (and thus its trace norm). We can also safely ignore all negative values of $\lambda$ other than $\lambda = -1$, since $C$ is a cone so we can absorb $|\lambda|$ into $X$.
    
    It follows that (b) is equivalent to $\|\rho - X\|_{\textup{tr}} \geq 1$ for all $X \in C$, which (by Equation~\eqref{eq:dist_defn}) is equivalent to $D_C(\rho) \geq 1$. Since $D_C(\rho) \leq 1$ always holds, this establishes the claim that (a) and (b) are equivalent.
    
    To see that (b) and (c) are equivalent, we recall from Theorem~\ref{thm:main_res} that $\rho$ is Birkhoff--James orthogonal to every $X \in C$ if and only if $\tr(XP) \leq \tr(X)/2$. Since this inequality is scale-invariant, without loss of generality we may set $\sigma = X/\tr(X)$ to see that (b) holds if and only if $\tr(P\sigma) \leq 1/2$ for every mixed quantum state $\sigma \in C$.
\end{proof}

Since $C$ is a convex cone, condition~(c) of Theorem~\ref{thm:mod_tr_dist_resource} can be weakened slightly to say that $\tr(P\sigma) \leq 1/2$ for every mixed quantum state $\sigma$ that lies on an extreme ray of $C$. For all of the cones $C$ that we will explore (i.e., the cones defining $k$-coherence or $k$-entanglement), these extreme quantum states are exactly the pure states (i.e., the rank-1 matrices). So for our purposes, the conditions of Theorem~\ref{thm:mod_tr_dist_resource} are furthermore equivalent to $\bra{v}P\ket{v} \leq 1/2$ for every pure quantum state $\ketbra{v}{v} \in C$.

\subsection{The Modified Trace Distance of Coherence}\label{sec:mod_tr_coh}

The modified trace distance of coherence was shown to be of limited use for pure states, in that it reaches its maximum value (i.e., $D_{\textup{coh}}(\ketbra{v}{v}) = 1$) on nearly all pure states (and therefore does not differentiate between these states) in \cite{johnston2018modified}. We now provide the natural generalization of this result to the case of arbitrary-rank (i.e., not necessarily pure) quantum states:

\begin{theorem}\label{thm:mod_tr_norm_range}
    Suppose $\rho \in M_n^{+}$ is a mixed quantum state and $P$ is the orthogonal projection onto $\mathrm{range}(\rho)$. Then $D_{\textup{coh}}(\rho) = 1$ if and only if $P_{j,j} \leq 1/2$ for all $1 \leq j \leq n$.
\end{theorem}

\begin{proof}
    This follows immediately from the equivalence of conditions (a) and (c) in Theorem~\ref{thm:mod_tr_dist_resource}: the only pure quantum states $\ketbra{v}{v} \in C$ come from choosing $\ket{v} = \mathbf{e_j}$ for some $1 \leq j \leq n$, so $D_{\textup{coh}}(\rho) = 1$ if and only if $P_{j,j} = \mathbf{e}_{\mathbf{j}}^*P\mathbf{e_j} \leq 1/2$ for all $1 \leq j \leq n$.
\end{proof}

Indeed, the above theorem is a direct generalization of \cite[Theorem~2]{johnston2018modified} from the rank-1 case to the general case (in the rank-1 case $\rho = \ketbra{v}{v}$, notice that $P_{j,j} = |v_j|^2$).

In addition to the condition of Theorem~\ref{thm:mod_tr_norm_range} being easy to check, it has the interesting consequence of bounding the possible rank of any state with $D_{\textup{coh}}(\rho) = 1$. Not only are low-rank states usually maximally coherent (in the sense of $D_{\textup{coh}}$), but high-rank states are \emph{never} maximally coherent:

\begin{corollary}\label{cor:mod_tr_norm_small_rank}
    Suppose $\rho \in M_n^{+}$ is a mixed quantum state with $D_{\textup{coh}}(\rho) = 1$. Then $\mathrm{rank}(\rho) \leq n/2$.
\end{corollary}

\begin{proof}
    If $D_{\textup{coh}}(\rho) = 1$ then Theorem~\ref{thm:mod_tr_norm_range} tells us that the projection $P$ onto $\mathrm{range}(\rho)$ satisfies $P_{j,j} \leq 1/2$ for all $j$, so $\tr(P) \leq n/2$. Since $\tr(P) = \mathrm{rank}(P) = \mathrm{rank}(\rho)$, the result follows.
\end{proof}

In particular, the above corollary tells us that all quantum states with maximal modified trace distance of coherence in the $n = 2$ and $n = 3$ cases are necessarily pure (i.e., rank-$1$). This is no longer true when $n \geq 4$, however.

\subsection{k-Coherence}\label{sec:k_coh}

As a generalization of coherence, for an integer $1 \leq k \leq n$, the resource theory of \emph{$k$-coherence} \cite{RBC18} makes use of the following set $C = C_{\textup{k-coh}}$ as its free states:
\begin{align*}
    C_{\textup{k-coh}} & \defeq \left\{ \sum_j \mathbf{v_j}\mathbf{v}_{\mathbf{j}}^* : \text{each} \ \mathbf{v_j} \in \mathbb{C}^n \ \text{has at most $k$ non-zero entries} \right\} \subseteq M_n^{+}.
\end{align*}

In the $k = 1$ case, this reduces down to exactly the set of incoherent states (i.e., $C_{\textup{1-coh}} = C_{\textup{coh}}$), whereas the $k = n$ case simply reduces to the set of \emph{all} mixed states (i.e., $C_{\textup{n-coh}} = M_n^{+}$). For intermediate values of $k$, these sets satisfy the obvious chain of inclusions
\[
    C_{\textup{coh}} = C_{\textup{1-coh}} \subsetneq C_{\textup{2-coh}} \subsetneq C_{\textup{3-coh}} \subsetneq \cdots \subsetneq C_{\textup{(n-1)-coh}} \subsetneq C_{\textup{n-coh}} = M_n^{+}.
\]

By using the fact that the pure states in $C_{\textup{k-coh}}$ are the unit vectors with at most $k$ non-zero entries, we immediately get a characterization of the mixed states that maximize the modified trace distance of $k$-coherence:

\begin{theorem}\label{thm:k_coh_equiv}
    Suppose $\rho \in M_n^{+}$ is a mixed quantum state and $P$ is the orthogonal projection onto $\mathrm{range}(\rho)$. Then $D_{\textup{k-coh}}(\rho) = 1$ if and only if the operator norm of every $k \times k$ principal submatrix of $P$ is at most $1/2$.
\end{theorem}

\begin{proof}
    This follows immediately from the equivalence of conditions (a) and (c) in Theorem~\ref{thm:mod_tr_dist_resource}: the pure quantum states $\ketbra{v}{v} \in C$ come from choosing the unit vectors $\ket{v} \in \mathbb{C}^n$ with at most $k$ non-zero entries, so $D_{\textup{coh}}(\rho) = 1$ if and only if $\bra{v}P\ket{v} \leq 1/2$ for all such vectors. By focusing on all the $\ket{v}$ that have their non-zero entries in particular locations, we see that this is equivalent to the principal submatrix of $P$ corresponding to those locations having operator norm at most $1/2$.
\end{proof}

Indeed, the above theorem reduces to Theorem~\ref{thm:mod_tr_norm_range} in the special case when $k = 1$, since the $1 \times 1$ principal submatrices of $P$ are exactly its diagonal entries. 

The following concept will be useful in what follows.  Let $\|A\|_{(k)}$ denote the maximal operator norm of a $k \times k$ principal submatrix of $A$ (for example, if $k = 1$ then $\|A\|_{(k)}$ is the maximal diagonal entry of $A$). For $k>1$, this is a norm (but that fact won't be important for the coming argument).  We will present a useful characterization of $\|P\|_{(k)}$ when $P$ is a projection.  First we need the following lemma.

\begin{lemma} If $P,Q\in M_n^{+}$ are orthogonal projections then $\| PQP  \|=\|QPQ \|$.   \end{lemma}

\begin{proof}  Since $PQP$ and $QPQ$ are both positive semidefinite matrices, their operator norms are equal to the their largest eigenvalues.  We note that it is well-known that if $A$ and $B$ are two square matrices, then $AB$ and $BA$ have the same eigenvalues.  Our result now follows from the fact that $PQP=PQ^2P=(PQ)(QP)$ and $QPQ=QP^2Q=(QP)(PQ)$.
\end{proof}

While this fact will not be used in the paper, we note that $\| PQP  \|$ has geometric significance being the square of the cosine of the first canonical angle between the range subspaces of $P$ and $Q$ (see \cite{li2015some,kribs2019isoclinic} for applications of canonical angles to quantum information).

\begin{lemma}\label{twoproj}  Let $P\in M_n$ be a projection and let $1\le k\le n$.  Then we have the equality
    
 \[
        \|P\|_{(k)} = \max \left\{ \sum_{j=1}^k\big|v_j^\downarrow\big|^2 : \ket{v} \in \mathrm{range}(P) \right\},
    \]
where $v_j^\downarrow$ is the entry of $\ket{v}$ with the $j$-th largest absolute value. 
\end{lemma}

\begin{proof}  Let $i_1, i_2,...,i_k$ be the rows of the $k \times k$ principal submatrix of $P$ with the largest operator norm.  Let $Q$ be the projection onto $\mathrm{span}\{\mathbf{e}_{i_1},\mathbf{e}_{i_2},\ldots,\mathbf{e}_{i_k}\}$.  Then $ \|P\|_{(k)} =\|QPQ \|=\| PQP  \|$.  It is well known that the operator norm of a Hermitian matrix is equal to its numerical radius (i.e $\| H\| = \max_{\| x\|=1} |\bra{x}H\ket{x}|$) and we apply this fact to $PQP$.  Let $\ket{v}$ be the unit vector which maximizes $\bra{x}PQP\ket{x}$.  If $\ket{v}\not \in \mathrm{range}(P)$, we can replace $\ket{x}=\ket{v}$ with $\ket{x}=P\ket{v} / \Vert P\ket{v}\Vert$ giving us an even larger value for $\bra{x}PQP\ket{x}$; hence  $\ket{v}\in \mathrm{range}(P)$.  Therefore  $ \|P\|_{(k)} =\bra{v}PQP\ket{v} =\bra{v}Q\ket{v}=\sum_{j=1}^k \big| v_{i_j} \big|^2=\sum_{j=1}^k\big|v_j^\downarrow\big|^2$.  The final equality follows from the fact that if  $\{ v_{i_j}\}_{j=1}^k$ were not the $k$ largest entries of $\ket{v}$ in absolute value, this would contradict our assumption that $i_1$, $i_2$, $\ldots$, $i_k$ are the rows of the $k \times k$ principal submatrix of $P$ with the largest operator norm.
\end{proof}

We use the above lemmas to get a rank bound that generalizes Corollary~\ref{cor:mod_tr_norm_small_rank} to this setting of $k$-coherence:

\begin{corollary}\label{cor:mod_ktr_norm_small_rank}
    Suppose $\rho \in M_n^{+}$ is a mixed quantum state with $D_{\textup{k-coh}}(\rho) = 1$. Then
    \[
        \rank(\rho) \leq \frac{n(n + 1 - 2k)}{2(n - k)}.
    \]
\end{corollary}

\begin{proof}
    Theorem~\ref{thm:k_coh_equiv} tells us that if $D_{\textup{k-coh}}(\rho) = 1$ then $\|P\|_{(k)} \leq 1/2$. On the other hand, we also have some lower bounds on $\|P\|_{(k)}$: it must be the case that $\|P\|_{(1)} \geq \tr(P)/n$, since otherwise the sum of the diagonal entries of $P$ would be strictly less than $n(\tr(P)/n) = \tr(P)$, which is a contradiction.
    
    For larger values of $k$, we claim that $\|P\|_{(k)} \geq \|P\|_{(1)} + (1 - \|P\|_{(1)})(k-1)/(n-1)$. If we let $\ket{w}$ be a particular $\ket{v}$ attaining the supremum for $\|P\|_{(1)}$ given in Lemma \ref{twoproj}, then we have
    \[
        \sum_{j=2}^n \big|w_j^\downarrow\big|^2 = \sum_{j=1}^n \big|w_j^\downarrow\big|^2 - \|P\|_{(1)} = 1 - \|P\|_{(1)}.
    \]
    Since the $\big|w_j^\downarrow\big|$'s are sorted in non-increasing order, this implies
    \[
        \sum_{j=2}^k \big|w_j^\downarrow\big|^2 \geq \frac{k-1}{n-1}(1 - \|P\|_{(1)}).
    \]
    Finally, this gives us
    \begin{align}\label{eq:p_k_ineq}
        \|P\|_{(k)} \geq \sum_{j=1}^k\big|w_j^\downarrow\big|^2 = \|P\|_{(1)} + \sum_{j=2}^k\big|w_j^\downarrow\big|^2 \geq \|P\|_{(1)} + \frac{k-1}{n-1}(1 - \|P\|_{(1)}).
    \end{align}
    If we substitute the inequality $\|P\|_{(1)} \geq \tr(P)/n = \rank(\rho)/n$ into Inequality~\eqref{eq:p_k_ineq}, and then use the fact that $1/2 \geq \|P\|_{(k)}$ and solve for $\rank(\rho)$, we get exactly the inequality given in the statement of the corollary.
\end{proof}

In the special case when $k = 1$, Corollary~\ref{cor:mod_ktr_norm_small_rank} says that $D_{\textup{coh}}(\rho) = 1$ implies $\rank(\rho) \leq n/2$, agreeing with Corollary~\ref{cor:mod_tr_norm_small_rank}. At the other extreme, if $k > n/2$ then Corollary~\ref{cor:mod_ktr_norm_small_rank} says that $D_{\textup{k-coh}}(\rho) = 1$ implies $\rank(\rho) \leq 0$, which of course is impossible. It follows that if $D_{\textup{k-coh}}(\rho) = 1$ then $k \leq n/2$. Furthermore, in the $k = n/2$ case if $D_{\textup{k-coh}}(\rho) = 1$ then $\rank(\rho) \leq 1$, so $\rho$ must be pure.

We end this subsection by making a (somewhat tangential) note that in the numerical linear algebra literature, there is a concept called the factor width of a matrix first defined in \cite{boman2005factor}.  That paper deals entirely with real matrices, so we give the natural extension of the definition to complex matrices.

\begin{definition}\cite[Definition 1]{boman2005factor}
  The factor width of a complex positive semidefinite matrix $A$ is the smallest integer $k$
for which there exists a complex (rectangular) matrix $V$ where $A = VV^*$ and each column
of $V$ contains at most $k$ non-zero entries.  
\end{definition}

It is easy to see that a density matrix is $k$-coherent if and only if it has factor width at most $k$.  In \cite{boman2005factor}, there is a characterization of matrices which have factor rank at most two. We first need the following definitions:

\begin{definition}\cite{varga1976recurring}  Let $A\in M_n$, then its comparison matrix (denoted $M(A)$) is an $n$ by $n$ matrix with $M(A)_{i,i}=\vert A_{i,i}\vert$ for all $i$ and $M(A)_{i,j}=-\vert A_{i,j}\vert$ whenever $i\neq j$.
  \end{definition}

\begin{definition} A square matrix is said to be an $H$-matrix if its comparison matrix is positive semidefinite. 
\end{definition}

We can now restate the main result of \cite{boman2005factor} in terms of $2$-coherent states.

\begin{theorem}\cite[Theorem 9]{boman2005factor}
    A density matrix is $2$-coherent if and only if it is an $H$-matrix.
\end{theorem}

\subsection{Entanglement and k-Entanglement}\label{sec:k_ent}

When we apply Theorem~\ref{thm:mod_tr_dist_resource} to the resource theory of entanglement (i.e., when we take $C = C_{\textup{ent}}$), we see that the modified trace distance of entanglement satisfies $D_{\textup{ent}}(\rho) = 1$ if and only if $\bra{v}P\ket{v} \leq 1/2$ for every pure \emph{product} state $\ket{v} = \ket{x} \otimes \ket{y} \in \mathbb{C}^m \otimes \mathbb{C}^n$, where $P$ is the orthogonal projection onto $\mathrm{range}(\rho)$. In the terminology of \cite{JK10,JK11}, this is equivalent to the statement that $\|P\|_{S(1)} \leq 1/2$, where $\|\cdot\|_{S(1)}$ is the norm defined by
\begin{align}\label{eq:s1_norm}
    \|X\|_{S(1)} & \defeq \max\big\{ (\bra{x} \otimes \bra{y})X(\ket{x} \otimes \ket{y}) : \ket{x} \in \mathbb{C}^m, \ket{y} \in \mathbb{C}^n \big\}.
\end{align}

A natural generalization of entanglement is the resource theory of \emph{$k$-entanglement}, in which the set $C = C_{\textup{k-ent}}$ of free states consists of states with \emph{Schmidt number} \cite{TH00} at most $k$ (where $1 \leq k \leq \min\{m,n\}$ is an integer):
\begin{align*}
    C_{\textup{k-ent}} & \defeq \left\{ \sum_j \mathbf{v_j}\mathbf{v}_{\mathbf{j}}^* : \text{there exist} \ \{\mathbf{x_{i,j}}\} \subset \mathbb{C}^m,  \{\mathbf{y_{i,j}}\} \subset \mathbb{C}^n \ \text{such that} \ \mathbf{v_j} = \sum_{i=1}^k \mathbf{x_{i,j}} \otimes \mathbf{y_{i,j}} \ \text{for all} \ j\right\}.
\end{align*}
By the spectral decomposition theorem, if $k = 1$ then we have $C_{\textup{1-ent}} = C_{\textup{ent}}$, whereas if $k = \min\{m,n\}$ then we have $C_{\textup{min\{m,n\}-ent}} = (M_m \otimes M_n)^{+}$. For intermediate values of $k$, these sets satisfy a chain of inclusions that is analogous to that of $k$-coherence:
\[
    C_{\textup{ent}} = C_{\textup{1-ent}} \subsetneq C_{\textup{2-ent}} \subsetneq C_{\textup{3-ent}} \subsetneq \cdots \subsetneq C_{\textup{(min\{m,n\}-1)-ent}} \subsetneq C_{\textup{min\{m,n\}-ent}} = (M_m \otimes M_n)^{+}.
\]

In this setting, the natural generalization of the norm defined in Equation~\eqref{eq:s1_norm} is
\begin{align*}
    \|X\|_{S(k)} & \defeq \max\left\{ \bra{v}X\ket{v} : \text{there exist} \ \{\mathbf{x_{i}}\} \subset \mathbb{C}^m, \{\mathbf{y_{i}}\} \subset \mathbb{C}^n \ \text{such that} \ \ket{v} = \sum_{i=1}^k \mathbf{x_{i}} \otimes \mathbf{y_{i}} \right\}.
\end{align*}
This quantity is also a norm, and an almost identical argument to the one from the start of this section demonstrates the following theorem:

\begin{theorem}\label{thm:k_ent_equiv}
    Suppose $\rho \in (M_m \otimes M_n)^{+}$ is a mixed quantum state and $P$ is the orthogonal projection onto $\mathrm{range}(\rho)$. Then $D_{\textup{k-ent}}(\rho) = 1$ if and only if $\|P\|_{S(k)} \leq 1/2$.
\end{theorem}

Unlike the case of coherence, the condition of Theorem~\ref{thm:k_ent_equiv} is difficult to check even in the $k = 1$ case. Indeed, computation of $\|\cdot\|_{S(1)}$ is NP-hard \cite{Gha10,Gur03} (which is expected, since computation of $D_{\textup{ent}}$ is also NP-hard). Despite this, numerous bounds on $\|P\|_{S(k)}$ are known, which let us derive rank bounds like the following:

\begin{corollary}\label{cor:mod_kent_norm_small_rank}
    Suppose $\rho \in (M_m \otimes M_n)^{+}$ is a mixed quantum state with $D_{\textup{k-ent}}(\rho) = 1$. Then
    \begin{align*}
        \mathrm{rank}(\rho) & \leq \frac{mn(\min\{m,n\} - 2k + 1)}{2(\min\{m,n\}-k)} \quad \text{and} \\
        \mathrm{rank}(\rho) & \leq \frac{(n + m + 2 - 4k)^2 - (n-m)^2}{4}.
    \end{align*}
\end{corollary}

\begin{proof}
    We use \cite[Theorem 4.15]{JK10}, which says that every orthogonal projection $P$ satisfies the pair of inequalities
    \begin{align*}
        \|P\|_{S(k)} & \geq \frac{(k-1)mn + (\min\{m,n\}-k){\rm rank}(P)}{mn(\min\{m,n\}-1)} \quad \text{and} \\
        \|P\|_{S(k)} & \geq \min\Big\{1,\frac{k}{\big\lceil \frac{1}{2}\big( n + m - \sqrt{(n-m)^2 + 4{\rm rank}(P) - 4} \big) \big\rceil}\Big\}.
    \end{align*}
    
    If $D_{\textup{k-ent}}(\rho) = 1$ then Theorem~\ref{thm:k_ent_equiv} tells us that $1/2 \geq \|P\|_{S(k)}$, which gives the claimed result after a bit of routine (but somewhat ugly) algebra.
\end{proof}

In particular, the first inequality of Corollary~\ref{cor:mod_kent_norm_small_rank} tells us that if $k > \min\{m,n\}/2$ then $\mathrm{rank}(\rho) \leq 0$, which impossible. It follows that there only exist quantum states $\rho$ with $D_{\textup{k-ent}}(\rho) = 1$ when $k \leq \min\{m,n\}/2$. Furthermore, in the extreme case when $k = \min\{m,n\}/2$, the second inequality of Corollary~\ref{cor:mod_kent_norm_small_rank} says that any state with $D_{\textup{k-ent}}(\rho) = 1$ must have $\mathrm{rank}(\rho) \leq |n - m + 1|$ (so if we also have $m = n$ then $\rho$ must be pure). At the other extreme, if $k = 1$ then the first inequality of Corollary~\ref{cor:mod_kent_norm_small_rank} is the stronger one, and it says that if $D_{\textup{ent}}(\rho) = 1$ then $\mathrm{rank}(\rho) \leq mn/2$.

\subsection{An Application to the NPPT Bound Entanglement Problem}\label{sec:nppt_be}

Given a quantum state $\rho \in (M_m \otimes M_n)^{+}$, one often wants to know if it can be \emph{distilled}: transformed into (an arbitrarily good approximation of) the \emph{pure maximally entangled} state
\[
    \ket{\phi_+} \defeq \frac{1}{\sqrt{\min\{m,n\}}}\left(\sum_{j=1}^{\min\{m,n\}} \mathbf{e_j} \otimes \mathbf{e_j}\right)
\]
via local operations \cite{HHH97}. In some cases where $\rho$ is undistillable, multiple tensor copies of it are distillable \cite{Wat04}, so we say that $\rho$ is \emph{$r$-copy (un)distillable} if $\rho^{\otimes r} \in (M_m^{\otimes r} \otimes M_n^{\otimes r})^{+}$ is (un)distillable.

It is straightforward to show that if $\rho$ is separable (i.e., if $\rho \in C_{\textup{ent}}$) then it is $r$-copy undistillable for all $r \geq 1$. More surprisingly, it has also been shown that there is a class of entangled (i.e., not separable) states called \emph{PPT states} that are $r$-copy undistillable for all $r \geq 1$ \cite{HHH98}. Entangled states with this undistillability property are called \emph{bound entangled}.

The \emph{NPPT bound entanglement} problem asks whether or not there are any other (i.e., non-PPT) states that are $r$-copy undistillable for all $r \geq 1$, and it is one of the central open questions in the theory of quantum entanglement \cite{HRZ20}. It has been shown that it suffices to answer this problem for a single-parameter family of states called \emph{Werner states} \cite{HH99}, and it has furthermore been conjectured that the specific Werner state
\begin{align}\label{eq:werner_state_defn}
    \rho_{2/n} \defeq \left(\frac{1}{n^2 - 2}\right)I - \frac{2}{n(n^2 - 2)}\left(\sum_{i,j=1}^{n} \mathbf{e}_{\mathbf{i}}\mathbf{e}_{\mathbf{j}}^* \otimes \mathbf{e}_{\mathbf{j}}\mathbf{e}_{\mathbf{i}}^*\right) \in (M_n \otimes M_n)^{+}
\end{align}
is bound entangled for all $n \geq 4$ \cite{DSSTT00}. Indeed, it is straightforward to check that $\rho_{2/n}$ is not PPT and that it is ($1$-copy) undistillable. However, even proving the simplest case of $2$-copy undistillability of this state when $n = 4$ remains elusive, despite significant effort \cite{CHS21,PPHH10}.

It was shown in \cite{JK10} that $\rho_{2/n}$ is $r$-copy undistillable if and only if the orthogonal projections defined recursively by
\begin{align}\label{eq:proj_def}\begin{split}
	P_{1} & \defeq \ketbra{\phi_+}{\phi_+} \in (M_n \otimes M_n)^{+} \quad \text{and} \\
	P_{r} & \defeq (I - P_{1}) \otimes P_{r-1} + P_{1} \otimes (I - P_{r-1}) \in (M_n^{\otimes r} \otimes M_n^{\otimes r})^{+} \quad \forall \, r \geq 2
\end{split}\end{align}
satisfy $\|P_r\|_{S(2)} \leq 1/2$. Here, it is understood that the partitioning of the $2r$ tensor factors used to compute $\|P_r\|_{S(2)}$ is the one in which the first factors of $P_1$ always ``stay together'', as do the second factors of $P_1$. That is, if we label the two tensor factors of $P_1$ as $P_{1} \in A \otimes B$ then we have $P_r \in (A^{\otimes r}) \otimes (B^{\otimes r})$, with the $S(2)$-norm being computed across the central tensor cut, \emph{not} $P_r \in (A \otimes B) \otimes \cdots \otimes (A \otimes B)$.

This connection with the $S(2)$-norm leads to the following characterization of the bound entanglement problem in terms of Birkhoff--James orthogonality and the modified trace distance of $2$-entanglement:

\begin{theorem}\label{thm:nppt_be_equiv}
    Let $\rho_{2/n} \in (M_n \otimes M_n)^{+}$ be the Werner state defined in Equation~\eqref{eq:werner_state_defn} and let $P_r$ be the orthogonal projection defined in Equation~\eqref{eq:proj_def}. The following are equivalent:
    \begin{enumerate}
        \item[(a)] $\rho_{2/n}$ is $r$-copy undistillable.
        
        \item[(b)] $P_r$ is Birkhoff--James orthogonal to every member of $C_{\textup{2-ent}}$.
        
        \item[(c)] There exists a quantum state $\sigma \in (M_n^{\otimes r} \otimes M_n^{\otimes r})^{+}$ with $\mathrm{range}(\sigma) = \mathrm{range}(P_r)$ and $D_{\textup{2-ent}}(\sigma) = 1$.
    \end{enumerate}
\end{theorem}

\begin{proof}
    By Theorem~\ref{thm:k_ent_equiv}, we know that condition (c) happens if and only if $\|P_r\|_{S(2)} \leq 1/2$, and we already discussed the fact that $\|P_r\|_{S(2)} \leq 1/2$ is equivalent to $\rho_{2/n}$ being $r$-copy undistillable \cite{JK10}. This shows that (c) is equivalent to (a).
    
    The fact that $\|P_r\|_{S(2)} \leq 1/2$ is also equivalent to (b) (and thus (a) and (c) are equivalent to (b)) follows from the equivalence of conditions (b) and (c) of Theorem~\ref{thm:mod_tr_dist_resource} when $C = C_{\textup{2-ent}}$.
\end{proof}

\section{Conclusions and Outlook}\label{sec:conclusions}

The notion of Birkhoff--James orthogonality is an extension of  the usual notion of orthogonality. We considered Birkhoff--James orthogonality in the setting of  $n\times n$ complex-valued matrices together with the trace norm. Our main result (Theorem~\ref{thm:main_res}) characterized when a Hermitian matrix $H$ is Birkhoff--James orthogonal in the trace norm to a given  positive semidefinite matrix $B$, via trace conditions involving the product of $B$ and the positive and negative eigenprojection matrices of $H$. These eigenprojection matrices are key to our study, and we used them to further characterize when $H$ is Birkhoff--James orthogonal in the trace norm to the set of all positive semidefinite diagonal matrices, and when $H$ is Birkhoff--James orthogonal in the trace norm to the set of all (not just positive semidefinite) diagonal matrices in certain special cases (e.g., when $H$ has a certain rank, when $H$ is positive semidefinite, and/or the dimension is small).

Our results have direct applications to quantum resource theories; we showed that the modified trace distance between any quantum state (positive semidefinite, trace-one matrix) $\rho$ and a given closed convex cone $C$ of positive $n\times n$ matrices equals $1$ (the maximum possible value) precisely when $\rho$ is Birkhoff--James orthogonal in the trace norm to every member of $C$. This allowed us to generalize a result from \cite{johnston2018modified}, which characterized when the modified trace distance of coherence of a given pure quantum state is maximal, in numerous ways: we now have a version of that result for non-pure states (i.e., matrices of higher rank), we now have a version of that result that applies to resource theories other than just quantum coherence, and we now have a bound on how large the rank of such states can be.

Our results, when applied to the resource theory of $2$-entanglement (i.e., Schmidt number $2$), provide an intriguing connection to the NPPT bound entanglement problem. While we do not solve that problem (it has been open for over a decade and is notoriously difficult), we showed that it can be phrased in terms of Birkhoff--James orthogonality, thus opening up a wide variety of new tools that can be used to explore it.

Our results highlight the clear utility of Birkhoff--James orthogonality. In terms of linear algebraic considerations, one open problem would be to extend the work from Hermitian to normal matrices, which we believe is possible but we avoided the added complexity of the problem as we were interested in the applications to quantum resource theory, which involves positive semidefinite matrices. It would also be interesting to exactly characterize when a given Hermitian matrix $H$ is Birkhoff--James orthogonal to every diagonal matrix, but our Examples~\ref{exam:higher_dim_bj_orth} and~\ref{exam:n5_rank2} suggest that such a characterization might be quite delicate and difficult to pin down.

\section*{Acknowledgements}
 N.J.\ was supported by NSERC Discovery Grant number RGPIN-2016-04003. R.P.\ was supported by NSERC Discovery Grant number 400550. S.P.\ was supported by NSERC Discovery Grant number 1174582, the Canada Foundation for Innovation (CFI) grant number 35711, and the Canada Research Chairs (CRC) Program grant number 231250. 
 
\bibliographystyle{alpha}
\bibliography{references}
\end{document}